\documentclass[11pt]{article}

\usepackage{times,fullpage,amsfonts,latexsym}
\usepackage[dvips]{graphics}
\usepackage{color}

\newcommand{\remove}[1]{}

\newtheorem{theorem}{Theorem}[section]
\newtheorem{definition}[theorem]{Definition}
\newtheorem{lemma}[theorem]{Lemma}
\newtheorem{proposition}[theorem]{Proposition}
\newtheorem{corollary}[theorem]{Corollary}
\newtheorem{claim}[theorem]{Claim}

\newtheorem{remk}[theorem]{Remark}
\newenvironment{remark}{\begin{remk} \begin{normalfont}}{\end{normalfont}
\end{remk}}

\newtheorem{transformation}[theorem]{Transformation}

\def\FullBox{\hbox{\vrule width 8pt height 8pt depth 0pt}}

\def\qed{\ifmmode\qquad\FullBox\else{\unskip\nobreak\hfil
\penalty50\hskip1em\null\nobreak\hfil\FullBox
\parfillskip=0pt\finalhyphendemerits=0\endgraf}\fi}

\def\qedsketch{\ifmmode\Box\else{\unskip\nobreak\hfil
\penalty50\hskip1em\null\nobreak\hfil$\Box$
\parfillskip=0pt\finalhyphendemerits=0\endgraf}\fi}

\newenvironment{proof}{\begin{trivlist} \item {\bf Proof:~~}}
  {\qed\end{trivlist}}

\newcommand{\defword}[1]{{\bf #1}}

\newcommand{\BatchSize}{{b}}
\newcommand{\Field}{\mathbb{F}}
\newcommand{\fs}{q}
\newcommand{\Block}{{n}}
\newcommand{\numz}{z}
\newcommand{\Cp}{C}
\newcommand{\Rate}{R}
\newcommand{\Aomn}{\cal A_{\hbox{omn}}}
\newcommand{\Asc}{\cal A_{\hbox{sc}}}
\newcommand{\Ass}{\cal A_{\hbox{rs}}}
\newcommand{\Abfo}{\cal A_{\hbox{co}}}
\newcommand{\Apk}{\cal A_{\hbox{pk}}}
\newcommand{\dummy}{{d}}
\newcommand{\hash}{{H}}
\newcommand{\parity}{{P}}
\newcommand{\rndm}{r}
\newcommand{\Dreal}{\cal D}
\newcommand{\Dsim}{\cal D'}
\newcommand{\Ebad}{\cal E_{\hbox{bad}}}
\newcommand{\errorprob}{e}

\begin{document}

\title{Adversarial Models and Resilient Schemes for Network Coding}

\author{Leah Nutman\thanks{ Computer Science Division,
The Open University of Israel, Raanana, 43107, Israel. {\tt
lnutman@gmail.com}} \and Michael Langberg\thanks{Computer Science
Division, The Open University of Israel, Raanana, 43107, Israel.
{\tt mikel@openu.ac.il}}}

\maketitle

\begin{abstract}
In a recent paper, Jaggi et al. \cite{JaggiLKHKM07}, presented a
distributed polynomial-time rate-optimal network-coding scheme
that works in the presence of Byzantine faults. We revisit their
adversarial models and augment them with three, arguably
realistic, models. In each of the models, we present a distributed
scheme that demonstrates the usefulness of the model. In
particular, all of the schemes obtain optimal rate $C-z$, where
$C$ is the network capacity and $z$ is a bound on the number of
links controlled by the adversary.
\end{abstract}

\section{Introduction} \label{intro:sec}
Network coding is a powerful paradigm for network communication. In
``traditional" networks, internal nodes simply transmit packets
that arrive to them (without any substantial change of their
content). In contrast, when performing network coding, internal
nodes of the network are allowed to mix the information from
different packets they receive before transmitting on outgoing
edges. This mixing may substantially improve the throughput of a
network, it can be done in a distributed manner with low
complexity, and is robust to packet losses and network failures,
e.g., \cite{ACLY00,LYC03,KM03, JagSanCEEJT, HoMKKESL06}.

The focus of this paper is network coding for multicast networks
(where a single sender wants to transmit the {\em same}
information to several receivers), at the presence of Byzantine
network faults. A Byzantine adversary that may maliciously
introduce erroneous messages into a network may be especially
disruptive when network coding is applied. The simple reason is
that any message (including the faulty ones) affect all messages
on its path to the recipient. Therefore, a single faulty message
may contaminate many more messages down the line.

Motivated by the above difficulty, there has been some work on
detecting and correcting Byzantine faults. We distinguish between
computationally unbounded and computationally bounded adversaries.
For computationally unbounded Byzantine adversaries, error
detection was first addressed in \cite{HoLKMEK:04}. This was
followed by the work of Cai and Yeung~\cite{YeuC:06a,CaiY:06a},
who generalize standard bounds on error-correcting codes to
networks, without providing any explicit algorithms for achieving
these bounds.  Jaggi et al.~\cite{JagLHE:05}, consider an
information-theoretically rate-optimal solution to Byzantine
attacks, which however requires a centralized design. Finally, a
distributed polynomial-time rate-optimal network-coding scheme was
recently obtained (independently) by Jaggi et al.
\cite{JaggiLKHKM07} and Koetter and Kschischang \cite{Koe:07}.
Error detection for multicast network coding in the presence of
computationally bounded Byzantine adversaries was also considered
in the past \cite{KroFM:04, GkaR:06, ChaJL:06}. In these works
various authentication schemes are performed at internal nodes of
the network. In \cite{KroFM:04, GkaR:06} a centralized trusted
authority is assumed to provide hashes of the original packets to
each node in the network, \cite{ChaJL:06} obviates the need for a
trusted entity under the assumption that the majority of packets
received at terminal nodes is uncorrupted.

This paper builds on the scheme of \cite{JaggiLKHKM07} to obtain distributed polynomial-time rate-optimal network-coding schemes in three realistic adversarial models. Our schemes, as well as those of \cite{JaggiLKHKM07} assume no knowledge of the topology of the network and follow the distributed network coding protocol of \cite{HoMKKESL06}. Namely, their implementation involves only a slight
modification of the source and destination while the internal
nodes can continue to use the standard protocol
of~\cite{HoMKKESL06}. Before we mention our contribution in
detail, we present a brief description of the adversarial models
studied in \cite{JaggiLKHKM07}. In the following informal summary,
a sender named `Alice' is interested in the transmission of
information to a group of receivers named `Bob' over a given
network. The Byzantine adversary, `Calvin', controls some of the
links of the network and injects erroneous messages into the
network in aim to corrupt the communication between Alice and Bob.

\noindent
{\bf Omniscient adversary model:} In this model Calvin is all-powerful and
all-knowing, and is limited only by the number of links $z$ under
his control. \cite{JaggiLKHKM07} obtained a network coding scheme
with optimal rate for this model of $C-2z$, where $C$ is the
network capacity.

\noindent
{\bf Secret channel model:} This model allows Alice to send to Bob a
short (low rate) secret, which is completely hidden from Calvin
(who is again all-powerful and all-knowing (excluding the secret), and is
limited by the number of links $z$ under his control).
\cite{JaggiLKHKM07} obtained a network coding scheme with optimal
rate of $C-z$ for this model. Notice that the rate achievable in
this model is strictly higher than that in the Omniscient model.
This secret channel model was originally referred to in
\cite{JaggiLKHKM07}  as the `shared secret model'. We rename it
here to secret channel model as the secret shared in this model
between Alice and Bob {\em may depend on Alice's message}. We
elaborate on this point in detail shortly.

\noindent
{\bf Limited eavesdropping model:}  The last model, which is the
least relevant to our work, limits the number of links on which
Calvin can eavesdrop (it was originally named ``limited adversary
model" and we rename it here for concreteness). In this model
\cite{JaggiLKHKM07,JaggiLan07} obtained a network coding scheme
with rate $C-z$, {\em as long as Calvin can eavesdrop on at most
$C-z$ links} (in addition to the $z$ links under his complete
control).

\subsection*{Our Contribution}

In this work we introduce three additional adversarial models, and
give optimal rate efficient distributed network-coding schemes in
each of the models. As mentioned above, our schemes (as well as
those of \cite{JaggiLKHKM07}) assume no knowledge of the topology
of the network and follow the distributed network coding protocol
of \cite{HoMKKESL06}. Roughly speaking, we obtain an optimal rate of $C-z$
on all the adversarial models described below (the optimality of our schemes follow, e.g., from \cite{JagLHE:05}).

\paragraph{Random-Secret Model:} The first model we present is
the random secret model in which Alice and Bob share a short
(uniformly distributed) random secret which is completely hidden from Calvin. Calvin is
all-powerful and all-knowing (excluding the secret), and is
limited by the number of links $z$ under his control. This model
differs from the `secret channel' model discussed  in
\cite{JaggiLKHKM07} in the sense that the secret that is shared by
Alice and Bob is not constructed as a function of  Alice's
message, but rather is uniformly distributed and independent of Alice's message.
The independence of the secret shared by Alice and Bob from the
actual message $M$ being transmitted by Alice has several
advantages. This allows Alice and Bob to share their secret {\em
prior} to the act of communicating $M$. For example, one may
consider the scenario where Alice and Bob are able to meet (or
communicate) in advance and share a large source of completely
random bits (such as a CD of uniformly generated bits). As long as
these bits are unknown to Calvin, they can be used overtime to
communicate at high rate over the network (without the need of an
additional low rate channel connecting Alice and Bob). Moreover,
as we will see shortly, communication in this model sets the
foundations for communicating at high rate in the setting in which
Calvin is computationally bounded (more specifically, in the {\em
symmetric key} cryptographic setting). We would like to note that
in the scheme of \cite{JaggiLKHKM07} for their `secret channel'
model, the secret information that Alice and Bob share indeed
strongly depends on the message $M$ Alice transmits to Bob and
hence cannot extend naturally to the examples mentioned above.

For the random secret model we obtain a network coding scheme with
optimal rate of $C-z$. Our scheme is obtained by a transformation
of the scheme of \cite{JaggiLKHKM07} for the secret channel model.
In our proof we do not need to get into the finer details of the
original scheme and instead observe and exploit a useful property
of the original secret composition.

\paragraph{Causal-Omniscient Model:}
As in the Omniscient model of \cite{JaggiLKHKM07}, in this model
we assume Calvin is all-powerful and all-knowing, and is again
limited by the number of links $z$ under his control. However, to
obtain rate greater than $C-2z$, we slightly restrict Calvin.
Namely, we assume that Calvin is causal. Specifically, when Calvin
injects messages into the network at time step $t$, he only has
access to messages sent by Alice at time steps at most $t+\Delta$.
Here, $\Delta$ is some parameter of the network which is
considered small compared with the length of the communication
stream.

We present an optimal rate distributed network coding scheme for
this model. Just as in the omniscient adversary model, our scheme
requires $C>2z$. However, in such a case we obtain a rate of
$C-z$ (compared with $C-2z$ in the omniscient adversary model).
Our scheme is obtained by a fully modular composition of two
network coding schemes from \cite{JaggiLKHKM07}: one for the omniscient adversary model and
the other for the secret channel model.
The study of the causal-omniscient model will set the foundations for communicating
at high rate in yet an additional setting in which Calvin is computationally
bounded, the {\em public key} cryptographic setting.

We would like to note that causal adversaries were also implicitly studied in
\cite{JagLHE:05} in the centralized setting, while in this work we
focus on the distributed setting. Nevertheless, the upper bounds
proven in \cite{JagLHE:05} imply that the requirement of $C>2z$ is
necessary (otherwise no information can be transmitted).

\paragraph{Computationally-Bounded Adversary Model:} While our
previous models did not make any computational assumptions on the
parties involved, we now turn to study the case in which Calvin is
computationally bounded (as before Calvin is all-knowing, excluding any secret keys, and limited by the number of links $z$ under his control).
In this setting, we present two results.
The first result uses the notion of symmetric key cryptography and
is based on our random secret model.
Roughly speaking,
in the case Calvin is computationally bounded, one may replace the
random secret in the random secret model, by a series of {\em pseudo-random} bits:
bits that would still look completely random (i.e., uniform) to Calvin. Now, to
generate an (effectively) unlimited amount of shared pseudo-random
bits, to be used in {\em several} executions of the random secret protocol, all Alice and Bob
need to do is exchange a {\em single} short secret key prior to
the communication process. This single key, and the bits it
generates may be used over essentially unlimited time to communicate at high
rate.

The second result addresses the public key cryptographic setting.
In this setting, each of the parties, Alice and Bob, hold a pair
of keys: a private key (known only to itself) and a public key
(known to all --- including Calvin). Encrypted point to point
communication between Alice and Bob can be done using these public
and private keys; {\em without} Alice and Bob ever meeting in
advance to exchange a shared secret key. However, in the model we
study, no point to point channel is available - and Alice would
like to communicate at high rate to Bob over a given network. We
present a network coding scheme for the model at hand. Our scheme
is based on the scheme we present for the causal-omniscience model,
with the sole difference that public-key encryption is used to
hide some of Alice's information from Calvin.

As common in the study of cryptographic primitives, both our
results are conditional --- in the sense that they hold assuming
that certain cryptographic primitives exist (such as the
assumption that factoring is hard). Under such assumptions, we
prove in the symmetric key setting that our scheme obtains an
optimal rate of $C-z$, and in the (weaker) public key setting we
obtain the same optimal rate under the condition that $C>2z$.
In this model Calvin is no longer causal, however, as
in the causal-omniscient model, it can be seen that the upper bounds of
\cite{JagLHE:05} imply that the latter requirement of $C>2z$ is
necessary.

We note that in the public key scenario, we assume that Alice
knows Bob's public key. For this reason, the public-key model
seems particularly suitable in settings where cryptography is
already involved (e.g., to ensure privacy and integrity of the
communication). In such a scenario, a public-key infrastructure
may already be available and computational limitations on the
adversary are usually already assumed.

The remainder of the paper is organized as follows. Section~\ref{sec:prelim} contains
the model definitions and notation.
Sections~\ref{sec:ss-model},~\ref{sec:bfo-model}~and~\ref{sec:cb-model}
discuss the three new models and schemes presented above.

\section{Preliminaries}
\label{sec:prelim} In this section we give the definitions and
notation that are required to model network coding in various
adversarial models. Our definitions and notation mostly
follow~\cite{JaggiLKHKM07}.

\paragraph{Network Model}
The network will be modelled as a graph. We assume our graphs are
acyclic, and the communication over them is done in a synchronous
manner. Namely, in each time step a single packet of information
can traverse an edge of the network.

\paragraph{Network-coding schemes}
We will consider the task of routing information over the network
from a single sender Alice to multiple receivers Bob (the setting
of multicast). In fact, in our analysis, it will usually be sufficient to consider
a single receiver Bob. The reason is that in the schemes we
suggest, neither Alice nor the network need to be aware of the
location of Bob in the network. Therefore, it will be
possible to extend each one of our schemes from the case of a
single receiver to the case of multiple receivers (this state of affairs is common in the study of multicast network coding, e.g. \cite{JaggiLKHKM07}).
We will therefore continue the
formalization, assuming a single receiver (and will address the
setting of multicast separately for each one of our schemes).

We will not assume that Alice, Bob or any other internal node is
aware of the network topology or of the location of Alice and Bob
in the network. The network topology will only influence the
maximal achievable rate. A network-coding scheme is defined by
Alice's \defword{encoder}, Bob's \defword{decoder}, and the
\defword{coding} performed in internal nodes. We will now discuss
those three components.

Let $M$ be the message Alice wishes to transmit to Bob. The
encoding algorithm of Alice adds some redundancy into the message,
thus obtaining an encoded message $X$. This information is routed
through the network, where it is further encoded (as a result of
the network coding). Bob receives encoded information that may
also encompass network faults. Bob's decoding algorithm, applied
to the encoded information, is supposed to factor out the network
faults and retrieve the original message $M$.

It is convenient to assume that Alice's encoded message $X$ is
represented by a $\BatchSize \times \Block$ matrix, where every
entry of the matrix is an element from a finite field
$\Field_{\fs}$.
We refer to a column as a \defword{slice}, to a
row as a \defword{packet}, and to each entry as a
\defword{symbol}. It is also useful to note that in all of the
schemes of~\cite{JaggiLKHKM07}, as well as ours, $X$ is composed
of the original message $M$, and in addition some $\delta\Block$
slices of redundancy. In other words, the size of $M$ is
$(1-\delta)$ that of $X$.

The specific coding performed by internal nodes is less relevant
to our work, as it is inherited without change
from~\cite{JaggiLKHKM07}. For concreteness, let us mention that
internal nodes, as well as Alice herself, perform random linear
network coding a la~\cite{HoMKKESL06}. Namely, for each of its
outgoing links, a node selects random coefficients of a linear
transformation over $\Field_{\fs}$ (the number of coefficients is
$\BatchSize$ for Alice and equals the indegree for any internal
node). The network coding of each of the slices of $X$ goes as
follows: First Alice sends on each of her out-going edges the
corresponding linear transformation of the symbols in the slice.
Whenever an internal node receives a symbol on each of its
incoming edges (which is in itself a linear transformation of the
slice's symbols), it sends on its outgoing edges the corresponding
transformation of those symbols. As common in the literature of network coding, as our graphs are acyclic, we assume that information from different
slices is not mixed throughout the communication process. This
can be established by sufficient memory at internal nodes of the
network.

\paragraph{Adversarial Model} Each one of the adversarial models
we consider in this paper is specified by the exact power of the
adversary Calvin. We mention here the common properties of Calvin.

Calvin has under his control $\numz$ network's links of his
choice.\footnote{The parameter $\numz$ represents Calvin's power.
It is possible to define $\numz$ as the min cut between Calvin's
links and Bob. This is at most, but may be strictly smaller than
the number of links under Calvin's control.} On these links Calvin
may inject his own packets, disguised as part of the information
flow from Alice to Bob. Calvin succeeds if Bob decodes a message
different than Alice's original $M$. The goal of the
network-coding scheme is to ensure this only happens with very
small probability while maximizing the rate in which information
flows from Alice to Bob.

We do not assume that Alice, Bob or any internal node are aware of
the links under Calvin's control. On the other hand, Calvin has
full knowledge of the network topology as well as the identity of
Alice and Bob. In all of our models we assume that Calvin has full
eavesdropping capabilities (i.e., Calvin can monitor the entire
communication on each one of the links). Calvin knows the encoding
and decoding schemes of Alice and Bob, and the network code
implemented by the internal nodes (including the random linear
coefficients). Furthermore, in our proofs, we assume that Calvin
selects the message $M$ that Alice transmits. This ensures
that our schemes work for {\bf every} message $M$ Alice sends to
Bob.

The \defword{network capacity}, denoted by $\Cp$, is the maximum
number of symbols that can be delivered on average, per time step,
from Alice to Bob, assuming no adversarial interference (i.e., the
max flow of information from Alice to Bob). The network capacity
is known to equal {\em the min-cut from Alice to Bob}. (For the
corresponding multicast case, $\Cp$ equals the minimum of the
min-cuts over all destinations.)
For a message $M$, the \defword{error probability $\errorprob(M)$}
is the probability that Bob reconstructs a message different from
Alice's message $M$. The (maximum) error probability of the
encoding scheme is defined to be $\errorprob=\max_M\{
\errorprob(M)\}$ (Here the maximization is taken over the message
$M$ of Alice).
The \defword{rate} is the number of {\em information} symbols that
can be delivered on average, per time step, from Alice to Bob . In
the parameters above, the rate equals $(1-\delta)\cdot\BatchSize$
(recall that $\delta$ is the fraction of redundant slices). Rate
$\Rate$ is said to be achievable if for any $\alpha>0$ and
$\epsilon>0$ there exists a coding scheme of block length $\Block$
with rate $\geq \Rate-\alpha$ and error probability $\errorprob
\leq \epsilon$.

\subsection{Building Blocks of our Schemes}
\label{sec:tools} Our network-coding schemes rely on the schemes
of~\cite{JaggiLKHKM07}, given in two adversarial models: the
omniscient adversary model and the secret channel model. We
discuss those schemes here.

\subsubsection{A scheme in the omniscient-adversary model}
In the omniscient adversary model, we put no restrictions on the
knowledge and ability of Calvin (see discussion in
Section~\ref{sec:prelim}). In this model,~\cite{JaggiLKHKM07} gave
a distributed polynomial-time scheme $\Aomn$, and proved for it
the following theorem:
\begin{theorem}[\cite{JaggiLKHKM07}]
$\Aomn$ achieves a rate of $\Cp - 2\numz$, in the
omniscient-adversary model, with code-complexity ${\cal
O}((\Block\Cp)^3)$. \label{thm:omn}
\end{theorem}

\subsubsection{A scheme in the secret-channel model}
In the secret channel model, we assume that Alice can secretly
send Bob a (short) message that is completely hidden from Calvin.
We put no additional restrictions on the knowledge and ability of
Calvin (see discussion in Section~\ref{sec:prelim}). In this
model,~\cite{JaggiLKHKM07} gave a distributed polynomial-time
scheme $\Asc$, and proved for it the following theorem:
\begin{theorem}[\cite{JaggiLKHKM07}]
$\Asc$ achieves a rate of $\Cp - \numz$, in the secret-channel
model, with code-complexity ${\cal O}(\Block\Cp^2)$. The
communication on the secret channel consists of at most $C^2+C$
symbols.\label{thm:sc}
\end{theorem}

In Section~\ref{sec:ss-model}, we give some more details on the
way the secret message is defined in $\Asc$.

\subsection{Proof techniques: reduction and worst case
analysis}


A scheme is said to be (information theoretically) secure against
an adversarial entity Calvin, if for any behavior of Calvin, Alice
is able to communicate her information to Bob (with high
probability). Loosely speaking, we think of Calvin as an
algorithmic procedure, which given certain inputs (such as the
network topology, Alice's information and the network code applied
by the network), computes which edges in the network to corrupt
and which error message to transmit.

There are several proof paradigms that can be used in an attempt
to establish the correctness of a given coding scheme. In this
work, the correctness of our coding schemes will be proven by
means of {\em reduction}. Namely, we build upon the results of
\cite{JaggiLKHKM07}, and prove that any adversarial entity Calvin that
breaks our schemes will imply an additional adversary (usually
referred to as Calvin') that will not allow communication in one
of the schemes presented in \cite{JaggiLKHKM07}.

More specifically, our proofs can be outlined as follows. We first
define our coding schemes. We will then assume for sake of
contradiction that they are not secure. As we would like our
schemes to be secure for {\em any} message $M$ of Alice, this will
imply the existence of an adversary Calvin that first chooses
which message $M$ Alice should send to Bob, and then is able to
corrupt the communication of $M$ between Alice and Bob. Thinking
of Calvin as an algorithmic procedure, we show how to define the
additional adversary Calvin' --- which is a procedure based on
Calvin. Finally we show that Calvin' is able to break one of the
(provably secure) schemes presented in \cite{JaggiLKHKM07} --- this
suffices to conclude our proof.

\section{Random-Secret Model}
\label{sec:ss-model} The random-secret model is similar to the
secret-channel model of \cite{JaggiLKHKM07} with the difference
that the secret information sent from Alice to Bob should be
random and independent of Alice's input message $M$. Formally, we
allow Alice to share with Bob a short secret (which is uniformly distributed). This secret will stay hidden from
Calvin.

\paragraph{Alice's secret and message encoding in $\Asc$}
Recall that $\Asc$ is the scheme presented in \cite{JaggiLKHKM07} for communication in the secret-channel model.
We will show how to transform $\Asc$ to a comparable scheme which
works in the random secret model. The only ingredients of $\Asc$
we need to recall is the structure of Alice's secret and of
Alice's encoder. The encoding of $M$ into $X$ is very simple: We
assume that Alice's message $M$ is a $\BatchSize \times
(\Block-\BatchSize)$ matrix over $\Field_\fs$. The matrix $X$ is
$M$ concatenated with the $\BatchSize \times \BatchSize$ identity
matrix, $I$. Namely, $X=[M\ I]$.

Alice's secret message is computed in two steps. She first chooses
$\Cp$ {\em parity symbols} uniformly at random from the field
$\Field_\fs$. The parity symbols are labelled $\rndm_\dummy$, for
$\dummy \in \{1, \ldots, \Cp\}$. We denote by $R$ the vector of
parity symbols. Corresponding to the parity symbols, Alice's {\em
parity-check matrix} $\parity$ is defined as the $\Block \times
\Cp$ matrix whose $(i,j)^{th}$ entry equals $(\rndm_j)^i$, i.e.,
$\rndm_j$ taken to the $i^{th}$ power. The second part of Alice's
secret message is the $\BatchSize\times \Cp$ {\em hash matrix}
$\hash$, computed as the matrix product $X\cdot \parity$. The
secret message sent by Alice to Bob on the secret channel is
composed of both $R$ and $\hash$. As $\Block \leq \Cp$, we indeed
have a secret of at most $C^2+C$ symbols.

\paragraph{A useful property of $\Asc$} Note that the vector $R$ of Alice's secret is
already uniform and independent of the message $M$. On the other
hand, the hash $\hash$ is a deterministic function of $R$ and $M$
(given by the equation $\hash = X\cdot \parity$). Our main
observation (which we will prove below) is the following: for
almost every value of $R$, when $M$ is uniform then $\hash$ is
uniform as well. Furthermore, it is enough that a small chunk of
$M$ will be uniform to guarantee the uniformity of $\hash$. This
suggests the following idea: instead of selecting $\hash$ as a
function of Alice's message, we can select both $R$ and $\hash$
uniformly at random. Later, Alice can tweak the message a bit such
that we indeed get $\hash = X\cdot
\parity$. We continue to formalizing this idea.

\subsection{Defining the new scheme $\Ass$}
We now show how to transform $\Asc$ into a scheme $\Ass$ with
comparable performance in the random secret model. To define the
scheme we now define the random secret, Alice's encoder, Bob's
decoding, and the coding in internal nodes.

\paragraph{The random secret} The secret shared between Alice and
Bob is composed of a length-$\Cp$ vector $R$ over $\Field_\fs$ and
a $\BatchSize\times \Cp$ matrix $H$ over $\Field_\fs$. Both are
selected uniformly at random (and independently of each other).

Even though $R$ and $H$ are selected uniformly, their function in
$\Ass$ is identical to the function of $R$ and $H$ in $\Asc$. We
therefore use the same notation as given above. In particular, we
refer to $H$ as the hash matrix. The elements of $R$ are referred
to as the parity symbols and denoted $\rndm_\dummy$, for $\dummy
\in \{1, \ldots, \Cp\}$. Furthermore, we define the corresponding
parity-check matrix $\parity$ as before.

\paragraph{Alice's encoder} We allow Alice to encode a slightly
shorter input message $M$ assumed to be $\BatchSize \times
(\Block-\BatchSize-\Cp)$ matrix over $\Field_\fs$. Alice encodes
$M$ into a $\BatchSize \times \Block$ matrix $X=[L\ M\ I]$, where
$L$ is a $\BatchSize \times \Cp$ matrix and $I$ is the $\BatchSize
\times \BatchSize$ identity matrix. The matrix $L$ is defined
(arbitrarily) such that $\hash=X\cdot\parity$. We show shortly
that this system of linear equations (on the elements of $L$) will
have a unique solution with high probability over $H$ and $P$. If this
system has no solution or more than a single solution we define $L$ arbitrarily (say, to be the
all-zero matrix).

\paragraph{Network coding and Bob's decoder}
Both the network coding and Bob's decoder are defined in the same
way as in $\Asc$ \cite{JaggiLKHKM07}. Once Bob decodes a matrix
$[\bar{L}\ \bar{M}]$, Bob discards of the $\BatchSize \times \Cp$
prefix $\bar{L}$ and outputs $\bar{M}$.

\subsection{Properties of $\Ass$}

We now state and prove the properties of $\Ass$ that are almost
identical to those of $\Asc$:

\begin{theorem}
$\Ass$ is a distributed polynomial-time scheme. $\Ass$ achieves a
rate of $\Cp - \numz$, in the random-secret model, with
code-complexity ${\cal O}(\Block\Cp^2)$. The random secret
consists of at most $C^2+C$ symbols.\label{thm:ss}
\end{theorem}

\begin{proof}
We will prove that the probability that Bob decodes correctly in
$\Ass$ is almost identical to the probability that Bob decodes
correctly in $\Asc$. The theorem will then follow immediately from
the definition of $\Ass$ and from Theorem~\ref{thm:sc}. We note
that even though Alice is able to send to Bob a little bit less
information in $\Ass$ than in $\Asc$ (specifically, Alice sends
$\BatchSize \cdot \Cp$ fewer elements of $\Field_\fs$), the rate
in both schemes is identical (as we consider the rate as $\Block$
goes to infinity).

Let us consider an adversary Calvin that makes $\Ass$ fail with
probability $\epsilon$. In particular, Calvin may chose a message
$M$ for Alice to send s.t.\ with probability $\epsilon$, Bob
reconstructs $\bar{M}$ which is different than $M$. We will define
an adversary Calvin' that makes $\Asc$ fail with probability
$\epsilon'\geq \epsilon-\Cp^2/\fs$. This will conclude our proof.

Calvin' is defined as follows.
First Calvin' imitates the message selection of Calvin (namely,
Calvin' uses the message $M'$ Calvin would have chosen given the
topology and the code of the network). If Calvin sets Alice's
input to the message $M$ then Calvin' sets Alice's input to
$M'=[L\ M]$, where $L$ is a uniformly chosen $\BatchSize \times
\Cp$ matrix. Then Calvin' continues to mimic Calvin, and behaves
identically (in particular Calvin' sends the same messages as
Calvin would on the same corrupted links).

As we see, Calvin' tries to fail $\Asc$ by mimicking an attack of
Calvin on the execution of $\Ass$. The success of Calvin's attack
on the execution of $\Ass$ depends both on the message $X=[L\ M\
I]$ transmitted by Alice and the secret information $R,H$ shared
by Alice and Bob. Let $\Dreal$ be the distribution over triplets
$(R,H,L)$ obtained when $R$ and $H$ are selected uniformly at
random (and independently of each other) and the matrix $L$ is
defined to satisfy $\hash=X\cdot\parity$ if a single  such $L$
exists, and is defined to be the all-zero matrix otherwise. Let
$A$ be the set of triplets $(R,H,L)$ on which Calvin's attack
succeeds (here we are assuming Calvin to be a deterministic
adversary, however our analysis extends naturally to the case in
which Calvin may act based on random decisions also). Namely, the
success probability of Calvin can be formalized as $\Pr[A]$, where
the probability is over the distribution $\Dreal$.

Now consider the success probability of Calvin' on $\Asc$ averaged
over messages of the form $M'=[L\ M]$ (where $L$ is chosen at
random). As before this probability depends on the message $X=[L\
M\ I]$ sent by Alice and by the information $R,H$ shared by Alice
and Bob. Let $\Dsim$ be the distribution over triplets $(R,H,L)$
obtained when $R$ and $L$ are selected uniformly at random (and
independently of each other) and $\hash$ is defined to be
$X\cdot\parity$. Recall that Calvin' mimics the behavior of
Calvin, thus Calvin' succeeds on the triplet $(R,H,L)$ iff Calvin
succeeds on $(R,H,L)$. Hence, the average success probability of
Calvin over messages of the form $M'=[L\ M]$ can be formalized as
$\Pr[A]$, where the probability is now over $\Dsim$. Notice that
the subset $A$ of triplets $(R,H,L)$ is the set used above in the
discussion on $\Ass$.

In what follows we show that $\Dreal$ and $\Dsim$ are almost
identical. This will suffice to prove our assertion.

\begin{definition}
The event $\Ebad$ on $R$ happens either if one of the parity
symbols is selected to be zero or if any two of the parity symbols
are identical. In other words, $\Ebad$ happens if there exists
$\dummy \in \{1, \ldots, \Cp\}$ such that $\rndm_\dummy=0$, or if
for two distinct $\dummy, \dummy' \in \{1, \ldots, \Cp\}$, we have
that $\rndm_\dummy=\rndm{_\dummy'}$.
\end{definition}

Note that $\Ebad$ is defined both for $\Dreal$ and for $\Dsim$. In
both cases, $R$ is uniformly distributed. Therefore
$\Pr_{\Dreal}[\Ebad]=\Pr_{\Dsim}[\Ebad]$. Furthermore, it is easy
to argue that this probability is at most $\Cp^2/\fs$ (simply,
each of the $\Cp$ parity symbols is zero or identical to a
previously selected parity symbol with probability at most
$\Cp/\fs$). We are now able to formalize our main observation:
\begin{lemma}
Conditioned on $\Ebad$ not happening, the two distributions
$\Dreal$ and $\Dsim$ are identical.
\end{lemma}
\begin{proof}(of lemma)
Let us fix any value of $M$. Let us also fix any value of $R$ such
that $\Ebad$ does not happen. We will show that conditioned on
every such fixings, the distributions $\Dreal$ and $\Dsim$ are
identical.

Let us decompose the $\Block \times \Cp$ parity-check matrix
$\parity$ into a $\Cp \times \Cp$ matrix $V$ and an $(\Block-\Cp)
\times \Cp$ matrix $\parity'$, such that $\parity = \left [
\begin{array}{c}
V \\
\parity'
\end{array}
\right ]$. By the definition of $\parity$, the matrix $V$ is the
Van der Monde matrix that corresponds to the parity symbols in
$R$. Since we assumed that $\Ebad$ does not happen, we have that
the parity symbols are all distinct and non zero. Therefore $V$ is
invertible.

With this notation, we can rewrite the equation $\hash=X\cdot
\parity$ as follows:
$$\hash=[L\ M\ I]\cdot \left [
\begin{array}{c}
V \\
\parity'
\end{array}
\right ]=L\cdot V + [M\ I]\cdot \parity'.$$ Since we already fixed
$M$ and $R$, we have that $[M\ I]\cdot \parity'$ is a fixed
matrix, which we will denote as $\hash'$. We also have that $V$ is
a fixed invertible matrix. We denote by $V^{-1}$ its inverse. Now
we have that $\hash=L\cdot V + \hash'$, or alternatively that
$L=(\hash-\hash')\cdot V^{-1}$. We can conclude that for every
value of $\hash$ there is exactly one value of $L$ for which
$\hash=X \cdot \parity$. We therefore have that the equation
$\hash=X \cdot \parity$ forces a one-to-one correspondence between
the values of $L$ and the values of $\hash$. Therefore, the
uniform distribution over $L$ induces the uniform distribution
over $\hash$ and vise versa. The lemma follows.
\end{proof}

Recall that we defined $A$ be the set of triplets $(R,H,L)$ on
which Calvin's attack succeeds. It follows from the lemma that
conditioned on $\Ebad$ not happening, $\Pr[A]$ is identical under
$\Dreal$ and $\Dsim$. Since we already argued that
$\Pr_{\Dreal}[\Ebad]=\Pr_{\Dsim}[\Ebad]\leq$ $\Cp^2/\fs$, we can
conclude that the probability that $\Ebad$ does not happen and $A$
does happen is at least $\epsilon-\Cp^2/\fs$ (regardless of
whether the probability is taken over $\Dreal$ or $\Dsim$). We can
finally conclude that Calvin' succeeds in failing $\Asc$ with
probability at least $\epsilon-\Cp^2/\fs$.
\end{proof}

\paragraph{The case of multicast} In the above description
of $\Ass$, we considered for simplicity the case of a single Bob.
In the setting of multicast, there are two possible scenarios.
First, it may be the case that Alice and each of the Bobs share
the {\em same} random secret. $\Ass$ extends to this scenario with
no change (simply because $\Ass$ completely ignores the location
of Bob in the network). We now address the more general scenario,
where Alice may share a different secret with each one of the
Bobs.

Our main observation is that almost all of the information Alice
transmits (the matrix $X$) is {\em independent of the random
secret}. The only part of $X$ that does depend on the secret is
the matrix $L$. This matrix is rather small and its size is
independent of the block-length $\Block$. Therefore to extend
$\Ass$ to the setting of multicast, all we need to do is to have
Alice send a different matrix $L_i$ for each of the secrets she
shares. Since the number of Bobs is bounded by the size of the
graph, this only results in negligible rate loss. To decode, each
one of the Bobs ignores the communication which relates to other
secrets and only keeps the communication related to his $L_i$. Bob
then decodes exactly as in $\Ass$.

It remains to argue that with high probability each one of the
Bobs will decode Alice's message $M$ correctly. Let us consider
Calvin's attempt to fail the receiver Bob whose secret corresponds
to the matrix $L_i$. Our previous analysis implies that each one
of the matrices $L_j$ for $j\neq i$ are with high probability
uniform and independent of both $L_i$ and $M$. Therefore, these
additional matrices cannot assist Calvin in the attempt to fail
this particular Bob. We conclude that each of the receivers will
decode correctly with high probability, and therefore all of them
are likely to decode correctly.

\begin{remark}
In the above we assumed that the secret shared between Alice and
each receiver Bob includes the index $i$, such that $L_i$
corresponds to their shared secret. It is possible to avoid this
assumption as follows: (1) Let the random secret between Alice and
Bob also contain a random (almost pair-wise independent) hash
function $g_i$. Alice augments the message $M$ with $g_i(M)$ for
all of those hash functions $g_i$. (2) Continue as before and have
Bob decode according to each of the $L_j$'s (as now we assume that
Bob does not know $i$ such that $L_i$ corresponds to his secret).
Some of these decodings may result in $\bar{M}\neq M$. But with
very high probability none of the erroneous decodings will be
authenticated by a correct hash $g_i(\bar{M})$ (as for every $M$
and $\bar{M}$ we have that $g_i(\bar{M})$ is almost uniform and
independent of $g_i(M)$).
\end{remark}

\section{Causal-Omniscient Model}
\label{sec:bfo-model}

Recall that in our model of communication, the columns of the
matrix $X$ (namely, each  slice of information from $X$) is
encoded independently over time. Given the network's latencies
(the number of steps it takes for a message to traverse the
network), we have that while an internal node $v$ sends messages
that correspond to the $t^{th}$ column of $X$, Alice may already
be sending messages that correspond to column $t'>t$. Therefore,
in the model in which Calvin can eavesdrop on all links, it
inherently has a ``pick into the future". Namely, when sending
messages which correspond to the $t^{th}$ column of $X$, we assume that Calvin knows all the columns of $X$ up to column
$t+\Delta$, where $\Delta$ is some fixed parameter of the network.
It is not hard to verify that $\Delta$ is at most the size of the
edge set $E$. However, it may be the case that Calvin does not
necessarily know any later columns of $X$. This motivates the
definition of the Causal-Omniscient model.

In the Causal-Omniscient model, Calvin has unlimited computational
power. He has under his control $\numz$ network links of his
choice. On these links Calvin may inject his own packets,
disguised as part of the information flow from Alice to Bob. We do
not assume that Alice, Bob or any internal nodes are aware of the
links under Calvin's control. On the other hand, Calvin has full
knowledge of the network topology as well as the identity of Alice
and Bob. Calvin has full eavesdropping capabilities (i.e., Calvin
can monitor the entire communication on each one of the links).
Calvin knows the encoding and decoding schemes of Alice and Bob,
and the network code implemented by the internal nodes (including
the random linear coefficients). Furthermore, we assume that
Calvin knows which message $M$ Alice is sending to Bob.

The only limitation on Calvin is the following: while Calvin is
allowed access to the internal state and randomness of all
parties, he does not get such access to Alice's state and
randomness. Note that such a limitation is implicit in all other
limited adversarial models considered here and in
\cite{JaggiLKHKM07}.\footnote{For example, in the random-secret
model, one has to hide the secret-key which is expressed in
various computations of both Alice and Bob.} The desired
implication of this limitation for the Causal-Omniscient model is
the following: let $\Delta$ be a fixed parameter of the network
that specifies a bound on the latency of the network. By the
discussion above, if Calvin's messages correspond to columns of
$X$ up to its $t^{th}$ column then we assume that all columns
beyond column number $t+\Delta$ are hidden from Calvin.

\subsection{The scheme $\Abfo$}

We now define the scheme $\Abfo$ for the Causal-Omniscient model.
The scheme is obtained by a completely modular composition of two
schemes: A scheme $\Asc$ in the secret-channel model, and a scheme
$\Aomn$ in the omniscient-adversary model. See more details on the
schemes in Section~\ref{sec:prelim}. The idea of the composition
is simple: first Alice, Bob (and the network) execute $\Asc$ with
Alice's input $M$, but {\em without Alice sending the message on
the secret channel} (simply because a secret channel is not
available in this model). Unfortunately, without the secret
message, Bob cannot decode $M$ correctly yet. Therefore, to
transmit this secret information, we suggest that Alice and Bob
execute $\Aomn$ with the secret message as Alice's new input.
Unfortunately, $\Aomn$ may reveal the secret message to Calvin as
well. Our simple observation is that as long as the secret message
is revealed {\em after the execution of $\Asc$ ends}, it is too
late for Calvin to cause any harm. Therefore, all that we need (so
that $\Abfo$ works) is for Alice to send $\Delta$ ``garbage"
columns between the executions of $\Asc$ and of $\Aomn$. We turn
to a formal definition of $\Abfo$:

\paragraph{Alice's encoder}
Alice invokes the encoding and secret generating algorithms of
$\Asc$ on her input $M$. Denote by $X_M$ the output of the
encoding and $S$ the message to be sent on the secret channel. Now
Alice invokes an independent execution of the encoding algorithm
$\Aomn$ on $S$ as input. Denote by $X_S$ the output of the
encoding. For reasons that will be made clear shortly, Alice
encodes a secret $S$ such that $X_S$ will be of block length
$n_S=(n/\Cp)^{1/3}$ (here $n$ is the block length of our scheme
and $C$ is the capacity). Recall, that the size of $S$ (and thus
the block length of $X_S$) in the secret channel scheme is
independent of $n$ and significantly smaller than $n_S$. Hence,
such a blowup in the size of $X_S$ can be obtained for example by
an arbitrary padding of $S$ with irrelevant information. As we
will see, this blowup will enable our scheme to have a low
probability of error (without significantly increasing the
code-complexity). As $n_S$ is much smaller than $n$, our rate
remains optimal. Alice's encoder now outputs $X=[X_M\ 0\ X_S]$,
where $0$ denotes the zero matrix with $\Delta$ columns.

\paragraph{Network coding} As in $\Asc$ and $\Aomn$, the network
coding is the standard random-linear coding of \cite{HoMKKESL06}.

\paragraph{Bob's decoding} Bob first uses the decoder of $\Aomn$ on
the suffix of the communication (which corresponds to the columns
of $X_S$). Denote by $\bar{S}$ the decoded message. Bob now
applies the decoder of $\Asc$ on the prefix of the communication
(which corresponds to the columns of $X_M$), with the (relevant
parts of the) secret message set to $\bar{S}$. Bob outputs the
decoded message, which we denote by $\bar{M}$.

\subsection{Properties of $\Abfo$}

We state the parameters obtained by $\Abfo$ in the following
theorem.

\begin{theorem}
$\Abfo$ is a distributed polynomial-time scheme. $\Abfo$ achieves
a rate of $\Cp - \numz$, as long as $\Cp>2\numz$, in the
Causal-Omniscient model, with code-complexity ${\cal
O}(\Block\Cp^2)$. \label{thm:bfo}
\end{theorem}

\begin{proof}
Most of the properties of $\Abfo$ follow from the related
properties of $\Asc$ and $\Aomn$, as given by
Theorems~\ref{thm:sc}~and~\ref{thm:omn}. The restriction that
$\Cp>2\numz$ guarantees positive rate for $\Aomn$ (as the rate of
$\Aomn$ is $\Cp-2\numz$). Other than that, $\Abfo$ inherits its
rate from $\Asc$. There is some loss of rate in $\Abfo$ (compared
with $\Asc$) due to the communication related to the zero columns
and to $X_S$. Nevertheless, this loss is negligible as $\Block$
tends to infinity. The choice of $n_S$ (the block length of
$X_S$), guarantees that the code complexity due to both building
blocks ($\Asc$ and $\Aomn$) will equal ${\cal O}(\Block\Cp^2)$.

It remains to bound the error probability $\epsilon$ of $\Abfo$.
Obviously, $\epsilon\leq \epsilon_1 + \epsilon_2$, where
$\epsilon_1$ is the probability that $\bar{S}\neq S$ while
$\epsilon_2$ is the probability that $\bar{S}=S$ but $\bar{M}\neq
M$. It follows that $\epsilon_1$ is bounded by the error
probability of $\Aomn$ when applied to messages of  block length
$n_S$ that corresponds to $X_S$. In \cite{JaggiLKHKM07} this error
is shown to be vanishing as the block length tends to infinity
(note that when $n$ tends to infinity so does
$n_S=(n/\Cp)^{1/3}$). To bound $\epsilon_2$ notice that Calvin
does not get access to $X_S$ until he is done corrupting $X_M$.
Thus $\epsilon_2$ is bounded by the error probability of $\Asc$
when applied to messages of  block length that correspond to
$X_M$. As the block length of $X_M$ is proportional to $n$, we
conclude our assertion.
\end{proof}

\begin{remark} While we described $\Abfo$ for the case of a
single Bob, it also applies with no change to the setting of
multicast (simply because $\Abfo$ completely ignores the location
of Bob in the network, and there is nothing that distinguishes one
Bob from the other).
\end{remark}

\section{Computationally-Bounded Adversary Model}
\label{sec:cb-model}

In this section we consider a limitation of a different flavor on
the strength of the adversary Calvin. Namely, we assume that
Calvin is {\em computationally} bounded. Assuming so allows us to
employ powerful cryptographic tools. The two results in this
section correspond to cryptographic tools that are applicable in
two different settings: (1) Symmetric-key cryptography (discussed
in Section~\ref{sec:sk-model}), and (2) Public-key cryptography
(discussed in Section~\ref{sec:pk-model}). As common in the study
of cryptographic primitives, both our results are conditional ---
in the sense that they hold assuming that certain cryptographic
primitives exist (such as the assumption that factoring is hard).

Note that apart from the computational limitations on Calvin his
powers are intact. In particular, Calvin has full eavesdropping
capabilities and has full knowledge of the network topology as
well as the identity of Alice and Bob. Calvin knows the encoding
and decoding schemes of Alice and Bob, and the network code
implemented by the internal nodes (including the random linear
coefficients).

\subsection{Symmetric-key cryptography} \label{sec:sk-model}

Recall our scheme $\Ass$ in the random secret model. Assuming that
Alice and Bob share a short random secret, this scheme allows them
to communicate a significant amount of information which is
specified by the block length $n$.  Unfortunately, Bob will only
be able to decode Alice's message after receiving the entire block
of communication. Therefore, it is natural to assume that every
time slot (e.g., every hour or every day, depending on the rate of
communication), Alice and Bob would like to terminate the previous
execution of $\Ass$ and to start a new execution. Let
$S_1,S_2,\ldots S_{\ell}$ be the sequence of secrets used in these
executions. In general, the secrets should be independent of each
other, which implies that Alice and Bob may need to share a long
secret if they communicate over a long period of time.

In the Computationally-Bounded Adversary model, Alice and Bob can
execute $\Ass$ many times while still only exchanging a short
random string $s$. For that purpose, Alice and Bob may use a
\defword{pseudorandom generator}. For a definition and thorough
discussion of pseudorandom generators see \cite{GoldreichI}.
Essentially, an efficiently computable function $G$ is a
pseudorandom generator, if (1) $G$ is length increasing (i.e., for
every input its output is longer than its input) and (2) $G(x)$ is
computationally indistinguishable from a uniform string, as long
as $x$ is uniformly distributed. In other words $G(x)$ is
effectively random. It is known that the existence of pseudorandom
generators is essentially the minimal cryptographic assumption (as
it is equivalent to the assumption that one-way functions exist).
A pseudorandom generator that expands the length of its input
implies a pseudorandom generator with arbitrary polynomial
expansion (again, the reader is referred to \cite{GoldreichI} for
a detailed discussion and references therein).

The relation to our context is now simple: Alice and Bob can
exchange a {\em single} short secret key $s$ prior to the
communication process. Applying a pseudorandom generator to this
single key, they can obtain many pseudorandom keys
$G(s)=S_1,S_2,\ldots,S_{\ell}$ to be used in repeated executions
of $\Ass$ (in fact, $\ell$ need not be known in advance as it is
possible to keep on expanding $G$'s output on the fly). The proof
that such repeated executions of $\Ass$ are still secure is rather
immediate from the definition of a pseudorandom generator and we
therefore only sketch it here. For any $i$, if $S_i$ is truly
random Calvin will fail the execution of $\Ass$ with very small
probability. Assume for the sake of contradiction that this is not
the case when $S_i$ is taken from the output of $G$. This gives a
way to distinguish the output of $G$ from random. The
distinguisher simply simulates the repeated execution of $\Ass$
(playing the roles of Alice, Bob, internal nodes, {\em and
Calvin}). Now if Calvin succeeds then the distinguisher can deduce
with non-negligible probability that $S_i$ is not truly random.

\begin{remark} As discussed in Section~\ref{sec:ss-model}, the
scheme $\Ass$ can be extended to the case of multicast. The idea
described here, of replacing the random keys in multiple
executions of $\Ass$ with pseudorandom keys, applies in the
setting of multicast as well (in that setting, Alice and each one
of the Bobs will share a short key that will be expanded to many
pseudorandom keys using the pseudorandom generator).
\end{remark}

\subsection{Public-key cryptography} \label{sec:pk-model}

A disadvantage of the scheme $\Ass$ is that Alice and Bob need to
share a common key. In this section we relax this set up
requirement and only ask that Bob holds a pair of keys: a private
key (known only to itself) and a public key (known to all
--- including Calvin). In such a setup, {\em without Alice and Bob ever
meeting or exchanging private information}, we are able to give a
network-coding scheme, $\Apk$, against a computationally-bounded
Calvin with very similar parameters to those of $\Abfo$ (which was
given in the Causal-Omniscient model).

\subsection{The scheme $\Apk$}

We present a network coding scheme, $\Apk$, for the public-key
model, that is very similar to our scheme for the Causal-Omniscient model. Again, we compose two schemes: A scheme
$\Asc$ in the secret-channel model, and a scheme $\Aomn$ in the
omniscient-adversary model. (See more details on the schemes in
Section~\ref{sec:prelim}.) Alice, Bob (and the network)
execute $\Asc$ with Alice's input $M$, but {\em without Alice
sending the secret message $S$ on the secret channel} (simply
because a secret channel is not available in this model). We would
like to execute $\Aomn$ with $S$ as Alice's new input.
Unfortunately, as in this model all of the information sent by
Alice (i.e., the matrix $X$ in its entirety) is known to Calvin
from the start, $S$ will be available to Calvin during the
execution of $\Asc$ and the scheme may fail. The solution is
simple: instead of sending $S$, Alice will first encrypt $S$ {\em
using Bob's public key}, and send the encryption to Bob using
$\Aomn$. We now describe our scheme and proof in more detail.

\paragraph{Public-key encryption}
A central ingredient in building $\Apk$ is a public-key encryption
(PKE) scheme. For a thorough discussion of public-key encryptions
see \cite{GoldreichII}. We describe here the relevant
definitions for completeness.

A \defword{PKE scheme} gives a way for two parties to communicate
securely even though they did not previously meet and exchange
secrets. The scheme is composed of three probabilistic polynomial
time algorithms -- the key generating algorithm $Gen$ and the
encryption and decryption algorithms $Enc$ and $Dec$: (1) The
input of $Gen$ is the security parameter $k$ (we will shortly
discuss the role of $k$), and its output is a pair of keys -- the
secret key $sk$ and the public key $pk$ (both are of length
polynomial in $k$). (2) The public key $pk$ (which is known to
everyone) is used for encryption. The encryption of a message $m$
is a ciphertext $y=Enc(pk,m)$. The plaintext $m$ may be of
arbitrary length and the length of $y$ is polynomial is the length
of $m$ and in $k$. (3) The secret key $sk$ allows decryption. For
every $y$ as above we have that $Dec(sk,y)=m$.

The security requirement from a PKE scheme is that a ciphertext
$y$ gives no information on the plaintext $m$ to a computationally
bounded adversary. More formally, we will use a PKE scheme which
is semantically-secure against chosen-plaintext attack (CPA)
\cite{GoldwasserMic84}. \footnote{In Remark~\ref{rem:CCA} we note
that in some cases one may choose to require security against
chosen-ciphertext attack (CCA security).} There are various
equivalent formalizations of this security requirement and the one
that seems most convenient for us is based on the notion of
indistinguishability. Loosely, this means that no efficient
adversary $Adv$ can distinguish an encryption of a message $m_0$
from an encryption of a message $m_1$, (where $Adv$ is also
allowed to select $m_0$ and $m_1$). In other words, given an
encryption $y=Enc(pk,m_{\sigma})$, where $\sigma$ is a uniformly
selected bit, an adversary cannot guess $\sigma$ with probability
significantly better than half. This is exactly where the security
parameter $k$ comes into play: the advantage over half of the
adversary in guessing $\sigma$ is smaller than $1/poly(k)$ for
every polynomial $poly$ (under stronger assumptions we may require
the advantage to be exponentially small in $k$).

 We are now ready to define $\Apk$ formally.

\paragraph{Security parameter} In this model, the network coding
scheme is defined per security parameter $k$. This parameter
should be chosen as to make the encryption scheme $\langle
Gen,Enc,Dec\rangle$ secure enough. As the errors we seek are of
the order of $1/\Block$, it is enough to take $k<\Block^{\alpha}$
for some small constant $\alpha>0$ (under stronger assumptions, $k$ may
even be logarithmic in $\Block$). This will imply that all the
ciphertexts used in our scheme are of negligible length compared
with $\Block$ (e.g.\ smaller than $\Block^{\alpha'}$ for any
$\alpha'>0$ of our choosing).

\paragraph{Bob's keys} Bob runs $Gen$ and gets as output the pair
$(sk,pk)$. Bob publishes $pk$ as his public key and saves his
secret key $sk$.

\paragraph{Alice's encoder}
Alice invokes the encoding and secret generating algorithms of
$\Asc$ on her input $M$. Denote by $X_M$ the output of the
encoding and $S$ the message to be sent on the secret channel.
Alice invokes $Enc(pk,S)$ and gets as output $Y$. Now Alice
invokes (using fresh randomness) the encoding algorithm of $\Aomn$
on $Y$ as input. Denote by $X_S$ the output of the encoding. Alice
 now outputs $X=[X_M\ X_S]$.
 As in Section~\ref{sec:bfo-model} we will pad $X_S$ if necessary
 to ensure that it has block length $n_S=(n/\Cp)^{1/3}$.

\paragraph{Network coding} As in $\Asc$ and $\Aomn$, the network
coding is the standard random-linear coding.

\paragraph{Bob's decoding} Bob first uses the decoder of $\Aomn$ on
the suffix of the communication (which corresponds to the columns
of $X_S$). Denote by $\bar{Y}$ the decoded message. Bob invokes
$Dec(sk,\bar{Y})$, and receives $\bar{S}$ as output. Bob now
applies the decoder of $\Asc$ on the prefix of the communication
(which corresponds to the columns of $X_M$), with the secret
message set to $\bar{S}$. Bob outputs the decoded message, which
we denote by $\bar{M}$.

\subsection{Properties of $\Apk$}

We state the parameters obtained by $\Apk$ in the following
theorem.

\begin{theorem}
$\Apk$ is a distributed polynomial-time scheme. $\Apk$ achieves a
rate of $\Cp - \numz$, as long as $\Cp>2\numz$, in the public-key
model, with code-complexity ${\cal O}(\Block\Cp^2)$.
\label{thm:pk}
\end{theorem}

\begin{proof}
Similar to the proof of Theorem~\ref{thm:bfo}, most of
the properties of $\Apk$ follow from the related
properties of $\Asc$ and $\Aomn$, as given by
Theorems~\ref{thm:sc}~and~\ref{thm:omn}.

It remains to bound the error probability $\epsilon$ of $\Apk$.
Obviously, $\epsilon\leq \epsilon_1 + \epsilon_2$, where
$\epsilon_1$ is the probability that $\bar{Y}\neq Y$ while
$\epsilon_2$ is the probability that $\bar{Y}=Y$ but $\bar{M}\neq
M$. As in the proof of Theorem~\ref{thm:bfo}, it is not hard to
argue that $\epsilon_1$ is bounded by the error probability of
$\Aomn$ (when applied to messages of  block length $n_S$). We
would now like to argue that $\epsilon_2$ is bounded by the error
probability of $\Asc$ (when applied to messages of  block length
corresponding to $X_M$). This will turn out to be correct up to a
negligible additional error (which relates to the security
property of the PKE scheme).

To bound $\epsilon_2$ we must argue that Calvin cannot use the
encryption $Y$ of $S$ to increase his probability of corrupting
the communication between Alice and Bob. Intuitively, this is
clear - as the encryption $S$ is computationally sound. However,
to prove our argument formally, we need to present our claim under
the sole assumption that our encryption is secure against a
chosen-plaintext attack.  To this end, we condition on the fact
that $\bar{Y}=Y$ (and hence Bob knows $S$) and show that a
successful Calvin in our setting will imply a successful Calvin in
an imaginary setting in which Bob is given $S$ (e.g., via a side
channel) and the value of $X_S$ transmitted over the network is
the encoding of an all zero message. This in turn, implies that
Calvin can corrupt the original secret-channel protocol $\Asc$ of
\cite{JaggiLKHKM07}, a contradiction. We now sketch the details.

Let Calvin be an adversary that causes $\bar{Y}=Y$ and
$\bar{M}\neq M$ in an execution of $\Apk$, with probability
$\epsilon_2$. As a mental experiment assume that Bob's decoder
receives $S$ as an additional input. Bob can then ignore $\bar{Y}$
and simply invoke the decoder of $\Asc$. In this experiment we
have that the probability that Calvin manages to cause $\bar{Y}=Y$
but $\bar{M}\neq M$ is still $\epsilon_2$. This follows
immediately from the properties of a PKE scheme (as if $\bar{Y}=Y$
we also have that $\bar{S}= S$).

Further revising this mental experiment, let us now assume that
Alice defines $Y$ as the encryption of the zero-message (or any
other fixed message), rather than the encryption of $S$. It is not
hard to argue that in this case, the probability that Calvin
causes $\bar{M}\neq M$ is at least $\epsilon_2-neg(k)$ where
$neg(\cdot)$ is some negligible function (that is asymptotically
smaller than $1/poly(\cdot)$ for every polynomial $poly(\cdot)$).
If this is not the case then we can easily devise an adversary
$Adv$ that breaks the security of the PKE scheme. $Adv$ will
simulate all parties of the network-coding scheme (Alice, Bob,
Calvin and the internal nodes). When Alice generates $S$ then
$Adv$ will set $m_1=S$ and will set $m_0$ to be the all-zero
message. $Adv$ then receives $y$ which is an encryption of one of
these messages. $Adv$ can now continue the simulation of the
network-coding scheme with $Y=y$. Finally, when the simulation is
over, $Adv$ will output one if $\bar{M}\neq M$ and zero otherwise.

Summing up, we have an adversary Calvin that causes $\bar{M}\neq
M$ with probability $\epsilon_2'=\epsilon_2-neg(k)$, in the
revised mental experiment based on $\Apk$. Note that in this
mental experiment $X_S$ is completely independent of $S$.
Therefore, it is possible to define an adversary Calvin' that
fails $\Asc$ with probability $\epsilon'_2$ by simulating the
attack of Calvin in the setting of the mental experiment.
The theorem therefore follows.
\end{proof}

\begin{remark}\label{rem:CCA}
In the definition of $\Apk$ we used a PKE scheme which is secure
against a chosen-plaintext attack. We proved that one invocation
of $\Apk$ works when the public-key is only used for this
invocation. In case the same public-key is used many times, and
especially if it used for messages sent from different senders it
may be safer to use a PKE scheme that is secure against a
chosen-ciphertext attack (see \cite{GoldreichII} for more
information).
\end{remark}

\paragraph{The case of multicast} In the above description
of $\Apk$, we considered for simplicity the case of a single Bob.
The scheme can be extended to the setting of multicast, in a
similar manner to the extension of $\Ass$ to multicast (see
discussion in Section~\ref{sec:ss-model}). In fact the extension
is a bit simpler in the case of $\Ass$ as we describe now.

We assume that the $i$'th receiver knows a pair consisting of a
secret key $sk_i$ and public key $pk_i$. The public key is known
to everyone including Alice. Now, for every $i$, Alice will
transmit (using $\Aomn$) the pair $(pk_i,Y_i)$, where
$Y_i=Enc(pk_i,S)$ (recall that in the basic scheme, Alice
transmits a single $Y$). Based on the properties of $\Aomn$ we can
assume that each of the Bobs correctly retrieves all of the pairs.
The $i$th receiver can decode $Y_i$ (which corresponds to its
public key $pk_i$) and continue the decoding of $M$ as in the
basic scheme. To argue that with high probability each one of the
Bobs will decode Alice's message $M$ correctly, we note that the
concatenation of all the different encryptions of $S$ still does
not reveal any information on $S$.

\section{Conclusions}
\label{sec:further}

In this paper we have introduced three adversarial models and have
argued that (1) The models may be realistic. (2) The models are
useful in the sense that they allow non-trivial improvements in
the parameters of network coding schemes. We feel that this calls
for more attention into the assumptions regarding adversarial
limitations and set-up assumptions that apply in ``real life"
scenarios. Are the models suggested here indeed applicable? Are
there any other realistic and useful models to consider?

\bibliographystyle{plain}
\bibliography{ncoding}

\end{document}